\documentclass{article}
\usepackage{xcolor,fullpage,amsmath,amsthm,amsfonts,amssymb,braket,hyperref,thmtools, thm-restate}
\usepackage{appendix}
\newtheorem{Lemma}{Lemma}
\newtheorem{claim}{Claim}
\newtheorem{Theorem}{Theorem}
\newtheorem*{Theorem*}{Theorem}
\newtheorem{Corollary}{Corollary}
\newtheorem{Proposition}{Proposition}
\newtheorem{Definition}{Definition}
\newtheorem*{Problem*}{}
\newtheorem{Fact}{Fact}

\newcounter{savedCSS}
\newcounter{hereCSS} 
\newcounter{savedGraph}
\newcounter{hereGraph} 
\newcounter{savedXGraph}
\newcounter{hereXGraph} 
\newcommand{\NP}{\mathsf{NP}}
\newcommand{\NPhard}{$\mathsf{NP}$-hard}
\newcommand{\defeq}{\stackrel{{def}}{=}}
\newcommand{\FF}{\mathbb{F}}

\newcommand{\MinDist}{\textsc{MinDist}}
\newcommand{\MultGapDist}{\textsc{MultGapDist}}
\newcommand{\AddGapDist}{\textsc{AddGapDist}}

\newcommand{\CC}{\mathbb{C}}

\newcommand{\CSSMinDist}{\textsc{CSSMinDist}}
\newcommand{\MultGapCSSDist}{\textsc{MultGapCSSDist}}
\newcommand{\AddGapCSSDist}{\textsc{AddGapCSSDist}}

\newcommand{\GraphMinDist}{\textsc{GraphMinDist}}
\newcommand{\MultGapGraphDist}{\textsc{MultGapGraphDist}}
\newcommand{\AddGapGraphDist}{\textsc{AddGapGraphDist}}

\newcommand{\MinDistDualDist}{\textsc{MinDistDualDist}}
\newcommand{\MultGapDistDualDist}{\textsc{MultGapMinDistDualDist}}

\newcommand{\AddGapHGPDist}{\textsc{AddGapHGPDist}}

\begin{document}
\title{On the hardness of approximating minimum distances of quantum codes}
\date{\today}
\author{
Elena Grigorescu\thanks{David R. Cheriton School of Computer Science, 
University of Waterloo, Canada. 
Supported in part by NSF CCF-2228814 while at Purdue University. \url{elena-g@uwaterloo.ca}} 
\and
Vatsal Jha\thanks{Department of Computer Science, 
Purdue University, IN. Supported in part by NSF CCF-2228814, NSF CCF-2330130, NSF CCF-2127806 and ONR Award N00014-24-1-2695. \url{jha36@purdue.edu}} 
\and
Eric Samperton\thanks{Departments of Mathematics and Computer Science, 
Purdue Quantum Science and Engineering Institute, 
Purdue University, IN. Supported in part by NSF CCF-2330130. \url{eric@purdue.edu}} 
}
\maketitle
\begin{abstract}

The problem of computing distances of error-correcting codes is fundamental in both the classical and quantum settings.  While hardness for the classical version of these problems has been known for some time (in both the exact and approximate settings), it was only recently that Kapshikar and Kundu showed the hardness of these problems for non-degenerate quantum codes. As our first main result, we get a direct reduction to CSS codes, the most commonly used type of quantum code, from the minimum distance problem for classical linear codes.  In particular, our proof uses arguably simpler constructions based on hypergraph product codes.

Our second set of results considers the distance of a graph state, which is a key parameter for quantum codes obtained via the codeword stabilized formalism.  We show that it is NP-hard to compute/approximate the distance of a graph state, which is a degenerate quantum code, when the adjacency matrix of the graph is the input. This establishes hardness for even degenerate quantum codes, which was not known before. Furthermore, we show this is true even if we only consider X-type errors of a graph state.  Our techniques, moreover, imply an interesting classical consequence: the hardness of computing or approximating the distance of classical codes with rate equal to 1/2.

One of the main motivations of the present work is a question raised by Kapshikar and Kundu concerning the NP-hardness of approximation when there is an additive error proportional to a quantum code's length.  We show that no such hardness can hold for hypergraph product codes. These observations suggest the possibility of a new kind of square root barrier.
\end{abstract}


\section{Introduction}
\label{sec:intro}

Achieving scalable fault-tolerant quantum computation is a  fundamental challenge in quantum computing. To circumvent this issue, the idea of \emph{encoding} quantum information of a smaller number of qubits into a larger number of qubits is a promising one. It was Shor's 9-qubit code \cite{Shor95decoherence} that provided the first example of a quantum error-correcting code that could correct arbitrary qubit errors on a single physical qubit.

This was followed by a sequence of works \cite{calderbankshor98},\cite{calderbank1998gf4},\cite{Gottesman1997stabilizer}, where a framework for constructing quantum error-correcting codes from classical error-correcting codes was developed. While these results are information-theoretic in nature, the setup of quantum error-correcting codes has also been studied from an algorithmic standpoint. It has been shown that both the problem of decoding quantum codes $\cite{Iyer15decodinghardness}$ and more recently, computing the error-correcting capability i.e., the  $\emph{distance}$, of (non-degenerate) quantum codes is $\NP$-hard \cite{kapshikar2023hardness}. In the current paper, we provide an alternate (and an arguably simpler) proof of the hardness of computing the distance of quantum codes from the one in \cite{kapshikar2023hardness} using the hypergraph product codes \cite{Tillich13}. Furthermore, we also establish the hardness of computing the distance for degenerate quantum codes.      

We now start by recalling the setup for both the classical and quantum minimum distance. 
\subsection{The Setup}

The decision problem corresponding to the (exact) minimum distance of a classical binary linear code takes as input a matrix $H$ with entries from the finite field of order two, i.e., $\FF_{2}:=\{0,1\}$ with integer addition and multiplication modulo $2$,  and an integer $t$. The $\mathrm{YES}$  instances of this decision problem are those $(H,t)$ pairs for which there exists a non-zero binary vector $x$ with Hamming weight $\leq t$ such that $Hx^{T}$  is the zero vector. While the hardness for computing the minimum distance of classical binary linear codes, both exactly \cite{vardy1997intractability} and approximately \cite{dumer2003approx,ChenggapMDP2012,AustrinKhot,Bhattiprolu2025PCPfree}, has been known for some time now, the hardness for quantum minimum distance, both exactly and approximately, was established recently \cite{kapshikar2023hardness}.   

For the quantum setting, the decision problem for the (exact) minimum distance of Calderbank-Shor-Steane (CSS) quantum codes takes as input a pair of matrices $H_{X}$, $H_{Z}$ and an integer t such that  
$H_{X}H_{Z}^{T}$ is the zero matrix.   For simplicity, let $C_{X}:=Ker(H_{X}),C_{Z}:=Ker(H_{Z})$ with $C_{X}^{\perp}$ and $C_{Z}^{\perp}$ as their orthogonal complement, respectively.  The $\mathrm{YES}$ instances are those inputs for which there exists a non-zero vector $(x,z)\in(C_{Z}\oplus C_{X})\setminus (C_{X}^{\perp}\oplus C_{Z}^{\perp})$ such that $H_{X}z^{T}$, $H_{Z}x^{T}$ are both the zero vectors and the Hamming weight of the bitwise $\mathrm{OR}$ of $x\lor z$ is at most $t$.

Looking at the setup of the (exact) classical and (exact) quantum minimum distance, it seems natural to reduce the former problem to the latter. This is not as easy as it seems for the primary reason that in the above setup there is no direct way of constructing CSS codes from arbitrary classical codes. To circumvent this, Kapshikar and Kundu considered the equivalent formalism of the codeword-stabilized (CWS) framework \cite{Cross2008Codeword,Bravyifermion2010}, where any stabilizer quantum code on $n$-qubits can be defined using a graph $G$ on $n$-vertices and a binary classical linear code $C$ of length $n$. Their reduction involved taking a linear code $C$ of length $n$ and for that code constructing an extremal graph on $n$ nodes that was 4-cycle free using the algebraic construction of \cite{Erdos1966Algebraic}. Now to obtain a CSS code from an arbitrary stabilizer code, one requires an extra reduction of \cite{Bravyifermion2010}. We recall Theorem 1 of \cite{kapshikar2023hardness} upto some extent: 

\begin{Theorem}[Informal]
    It is $\NP$-complete to compute the exact minimum distance of CSS codes.
    Additively or multiplicatively approximating the distance of CSS codes is
    also $\NP$-hard. The hardness of the above problems remains under the
    promise that the CSS codes are non-degenerate.
\end{Theorem}

In the current work, we instead use the hypergraph product framework of directly constructing CSS codes from an arbitrary pair of classical linear codes \cite{Tillich13}. The interesting thing about the hypergraph product construction is that under specific assumptions, one can obtain an exact relation between the obtained quantum code and the constituent classical linear codes, which is precisely the idea exploited by our reduction.
Before stating our results more precisely, we note that the HGP construction is an important technique, as it was the first construction of quantum codes that attains a square root in length distance at a constant rate.
This construction was thus an important precursor to a long line of works that culminated in the proof of existence of ``good'' quantum LDPC codes \cite{Bravyi2014homological}, \cite{Kaufman2021cosystolic} and \cite{Nicolas2021Union}, and so it is desirable to have direct methods for exhibiting the intrinsic hardness of problems involving HGP codes. The only limitation is that nothing can be concluded about the
(non-)degeneracy of the resulting CSS codes. 

So without commenting on the degeneracy of the obtained quantum codes, our result on the hardness can be stated as:
\begin{Theorem}[Informal]\label{thm:CSSdist-informal}
    It is $\NP$-complete to compute the exact minimum distance of CSS codes.
    Additively or multiplicatively approximating the distance of CSS codes is
    also $\NP$-hard. 
\end{Theorem}
\setcounter{savedCSS}{\value{Theorem}}

Another goal of the present work is to study the hardness of computing the distance of stabilizer states, i.e., quantum codes encoding zero logical codes. Specifcally, we consider the graph state \cite{Hein2004Multiparty,Damian2008graph,Cabello2011Optimalgraphstate} associated with a graph $G$ and its associated minimum graph state distance. The  \emph{minimum graph state distance} is defined as:
\[d_G \defeq \min\{wt_{H}(x\lor z):A_{G}x^{T}=z^{T}, \; (x,z)\neq O_{2n\times 1}\},\] where $A_{G}$ is the adjacency matrix of a graph. More generally, graph states have applications in quantum secret-sharing, quantum metrology and even quantum error-correction.
They are extensively studied in \cite{Damian2008graph}, \cite{Hein2004Multiparty}, \cite{Cabello2011Optimalgraphstate}, for example. Furthermore, as any stabilizer can be interpreted as a degenerate stabilizer code, it becomes natural to consider the distance problem for graph states.

\begin{Theorem}[Informal]\label{thm:GraphDist-informal}
    It is $\NP$-complete to compute the minimum distance of graph states. The
    hardness remains even for multiplicatively approximating by every constant factor and additively approximating upto a cube-root term in the length of the graph state.
\end{Theorem}
\setcounter{savedGraph}{\value{Theorem}}

To prove the above theorem, we consider a problem in classical coding theory, which we refer to as the $\MinDistDualDist$ problem, and then prove its hardness. The latter problem seems to be of independent interest as it also provides the hardness of computing the distance of classical codes with rate equal to $1/2$.

Interestingly, we show that even computing the promise version of $\GraphMinDist$ where $z=O_{1\times n}$ is $\mathsf{NP}$-hard.
More precisely, define the \emph{minimum graph state $X$-distance $d_{X,G}$} by
\[d_{X,G} \defeq \min\{wt_{H}(x): x\in \FF_{2}^{n}\setminus O_{1\times n}, A_{G}x^{T}=O_{n\times 1}\}.\]

 \begin{Theorem}[Informal]\label{thm:XGraphDist-informal}
    It is $\NP$-complete to compute the minimum $X$-distance of graph states.
    The hardness remains when multiplicatively approximating by any constant factor or additively approximating upto a linear term in the length of the graph state.
    
\end{Theorem}
\setcounter{savedXGraph}{\value{Theorem}}

One motivation for considering the $\emph{\GraphMinDist}^X$ arises naturally from the study of graph states transmitted through a bit-flip channel. Since a quantum bit-flip channel introduces only $X$-errors on each qubit, any correctable error must be counteracted by $Z$-type operators. Consequently, the minimum graph state distance in this setting is precisely the minimum graph state $X$-distance. Furthermore, since the minimum graph state $X$-distance serves as a natural upper bound on the minimum graph state distance, computing this quantity is itself a well-motivated problem.

As mentioned in the statement of the above problems also have approximate versions with multiplicative as well as additive factors which was also discussed in \cite{kapshikar2023hardness}. We show the hardness for the constant multiplicative factors and certain settings of additive factors.  Interestingly, both our constructions and those in the reduction of \cite{kapshikar2023hardness} obtain a square-root additive error for the minimum distance of CSS codes, hence hinting at a new ``square-root" barrier.

\subsection{Subsequent Work}
In recent manuscripts \cite{Ostrovsky2026Metric} and \cite{kapshikar2026hardness}, the problem of additively approximating the minimum distance of quantum codes upto a linear term in the length of the quantum code is claimed to be \NPhard. Furthermore, \cite{Ostrovsky2026Metric} also claims that the additive approximation of the minimum distance of graph states upto a linear term in the length of the quantum code is \NPhard. These works answer our question on the square-root barrier in \ref{sec:outlook}.

\subsection{Organization}
Section \ref{Preliminaries} briefly reviews all of the necessary background about classical and quantum codes.
We prove Theorem \ref{thm:CSSdist} in Section \ref{sec:thm1}.
Theorems \ref{Graph-State-Distance} and \ref{X-Graphstatedistance} are proved in Subsections \ref{ss:GSD} and \ref{ss:XGSD}, respectively.
The final Section \ref{sec:outlook} includes extensive discussion, especially concerning the matter of hardness for $\AddGapCSSDist_{\tau,\epsilon}$ when $\epsilon>\frac{1}{2}$.

\section{Preliminaries}
\label{Preliminaries}
\subsection{Classical codes and their distance problems}

A \emph{(classical, binary) code} $C$ of \emph{length} $n$ is any subset $C \subseteq \FF_2^{n}$, where $\FF_2=\{0,1\}$ is the binary field and $\FF_2^n$ the $n$-dimensional $\FF_2$-vector space consisting of all row vectors of length $n$.
If $C \le \FF_2^n$ is a linear subspace, then we call $C$ a \emph{linear} code.
All classical codes considered in this paper will be linear, and so we will often not use the word ``linear'' when we should.
The two most important parameters of a classical code are its \emph{dimension} $k \defeq \dim_{\FF_2} C$ and its \emph{distance}

\[d \defeq \min\{d_H(x,y) \mid x,y\in C, x\neq y\} = \min\{d_H(x,O_{1\times n}) \mid x \in C,x\neq O_{1\times n}\},\]
where $d_H(x,y) \defeq |\{ i\in[n] \mid x_{i}\neq y_{i}\}|$
is the Hamming distance between vectors in $\FF_{2}^n$ and $wt_{H}(x)=d_H(x,O_{1\times n})$ is the Hamming weight of a vector $x\in\FF_{2}^{n}$ .
A classical linear code with such parameters is called an $[n,k,d]$ code.

There are two standard ways to present a classical linear code: with a \emph{generator matrix} $G$, or with a \emph{parity-check matrix} $H$.
The former is any binary matrix whose row space equals $C$,
while the latter is any matrix whose kernel equals $C$---or, more precisely (due to the difference between row and column vectors),
\[ C = \{ x \in \FF_2^n \mid Hx^T = O\}.\]
For a given parity-check matrix $H$, we let $C(H)$ denote the code corresponding to $H$.
Parity-check matrices can be converted to generator matrices, and \emph{vice versa}, in polynomial time.
Moreover, when convenient, we may assume without loss of generality that the parity check matrix $H$ of a $[n,k,d]$ code is a $(n-k) \times n$ matrix (in particular, has full rank $n-k$). A natural way of specifying a parity-check matrix $H$ is via the systematic form where $H=[I_{n-k}:M_{(n-k)\times k}]$.\footnote{For matrices $A$ and $B$ with equal number of rows, [A:B] refers to their augmented matrix}
Given the parity-check matrix in systematic form, $G=[M^{T}_{k\times (n-k)}:I_{k}]$ will be a generator matrix for $C$.

The problem of calculating the distance of a classical code from its parity-check matrix has the following standard decision variant.
\begin{quote}
\label{MinDist}
\underline{Classical Minimum Distance Decision Problem $(\MinDist)$}

\textbf{Instance:} A binary matrix $H\in\FF_2^{(n-k)\times n}$ and a non-negative integer $t$.

\textbf{Output:} YES if $d\leq t$.  NO otherwise.
\end{quote}
There are also important approximation variants of $\MinDist$, either with a multiplicative error or an additive error.
The decision variants of these are most conveniently formulated as promise problems with a gap.
Let $\gamma\ge 1$ and $\tau>0$.
\begin{quote}
\label{MultGapDist}
\underline{Classical Minimum Distance Decision Problem with Multiplicative Gap $\gamma$}

\underline{($\MultGapDist_\gamma$)} 

\textbf{Instance:}  A binary matrix $H\in\FF_2^{n-k\times n}$ and a non-negative integer $t$.

\textbf{Promise:} Either $d\leq t$ or $d>\gamma t$.
    
\textbf{Output:} YES if $d\leq t$. NO if $d>\gamma t$.
\end{quote}
\begin{quote}
\label{AddGapDist}
\underline{Classical Minimum Distance Decision Problem with Additive Gap $\tau n$}

\underline{($\AddGapDist_\tau$)}

\textbf{Instance:}  A binary matrix $H\in\FF_2^{n-k\times n}$ and a non-negative integer $t$.

\textbf{Promise:} Either $d\leq t$ or $d>t+\tau n$.

\textbf{Output:} YES if $d\leq t$. NO if $d>t+\tau n$.
\end{quote}


We need one more notion from classical codes before we can properly introduce CSS codes.
Given a classical code $C$ of length $n$, the \emph{dual code} $C^\perp$ is another length $n$ code defined as follows:
\[C^\perp \defeq \{c'\in\FF_2^n| \langle c',c\rangle=0, \text{ for all } c\in C \}\]
where $\langle c',c\rangle$ is the usual $\FF_2$ dot product.
It is easy to see that a generator matrix for $C$ is a parity-check matrix for $C^{\perp}$.
In particular, if $H$ is any parity-check matrix of a code $C$ and $G$ is a generator matrix for $C$, then $HG^{T}=O$.

\subsection{Quantum codes and their distance problems}
We begin this subsection by reviewing the framework of quantum error-correcting codes, together with the stabilizer formalism for quantum codes developed in \cite{Gottesman1997stabilizer} and \cite{calderbank1998gf4}.
We conclude by discussing graph states, CSS codes, and the precise decision problems associated to their distances that we study.

A single qubit lives in the two-dimensional Hilbert space $\mathbb{C}^{2}$, while a system of $n$ qubits lives in the $n$-fold tensor product $(\mathbb{C}^{2})^{\otimes n}$.
Since the state space of one qubit is two-dimensional, an arbitrary single-qubit state $\ket{\psi}$ may be written as
$\ket{\psi}=\alpha\ket{0}+\beta\ket{1}$, with $|\alpha|^{2}+|\beta|^{2}=1$, where $\ket{0}=\begin{bmatrix}
    1\\
    0
\end{bmatrix}$ and $\ket{1}=\begin{bmatrix}
    0\\
    1
\end{bmatrix}$.
Next, we turn to the systems of several qubits and the operators acting on systems of several qubits. As a first step, we recall how the tensor product of two matrices is defined.
\begin{Definition}
    Given two matrices $$A_{n\times m}:=\begin{bmatrix}
        a_{11} & a_{12} &. &. &. &a_{1m}\\
        a_{21} & a_{22} &. &. &. & a_{2m}\\
               & &.& & &\\
               & &.& & &\\
               & &.& & &\\
        a_{n1} & a_{n2} &. &. &. & a_{nm}\\
    \end{bmatrix}$$
    and
    $$B_{p\times q}:=\begin{bmatrix}
        b_{11} & b_{12} &. &. &. &b_{1q}\\
        b_{21} & b_{22} &. &. &. & b_{2q}\\
               & &.& & &\\
               & &.& & &\\
               & &.& & &\\
        b_{p1} & b_{p2} &. &. &. & b_{pq}\\
    \end{bmatrix},$$
\end{Definition}
their tensor product is the $(np)\times(mq)$ block matrix
$$A\otimes B:=\begin{bmatrix}
        a_{11}B & a_{12}B &. &. &. &a_{1m}B\\
        a_{21}B & a_{22}B &. &. &. & a_{2m}B\\
               & &.& & &\\
               & &.& & &\\
               & &.& & &\\
        a_{n1}B & a_{n2}B &. &. &. & a_{nm}B\\
    \end{bmatrix}.$$

An $n$-qubit system has state space $(\mathbb{C}^{2})^{\otimes n}$.  Hence, a general $n$-qubit state $\ket{\psi}$ takes the form
\[\ket{\psi}=\underset{x\in\FF_{2}^{n}}{\sum}\lambda_{x}\ket{x},\]
where the coefficients satisfy $\underset{{x\in\FF_{2}^{n}}}{\sum}|\lambda_{x}|^{2}=1$.

A {\em quantum error correcting code of length $n$} is simply a $\CC$-linear subspace $Q$ of the $n$-qubit Hilbert space $(\CC^2)^{\otimes n}$.\footnote{
In this naive sense, every quantum error-correcting code is ``linear.''
The distinction that actually matters in the quantum setting, however, is the one between ``additive'' and ``non-additive'' codes:
additive quantum codes play the role of linear classical codes, while non-additive quantum codes correspond to non-linear classical codes.
We do not elaborate on this distinction here, beyond remarking that every Pauli stabilizer code is additive, and that CSS codes form an important subclass of Pauli stabilizer codes.
}
By analogy with the classical setting, the \emph{minimum distance $d_Q$ of a quantum error-correcting code} is defined as the smallest weight of a non-trivial physical error that induces a non-trivial logical error on the code space; a precise definition is given later. A quantum code of length $n$ with $\dim_\CC Q = K$ and distance $d_{Q}$ is called an $((n,K,d_{Q}))$ quantum code. If $K = 2^k$ for some non-negative integer $k$, the code is said to encode $k$ \emph{logical qubits}, and is written as a $[[n,k,d_{Q}]]$ quantum code.
In this paper, our attention is restricted to a special family of quantum codes known as stabilizer quantum codes. We begin by reviewing the formalism underlying these codes.

\begin{Definition}\label{Pauli-Group}
    The Pauli group on $n$ qubits is defined to be
    \[\mathcal{P}_{n}=\{i^{\lambda}M_{1}\otimes M_{2}\otimes...\otimes M_{n}:\lambda\in\{0,1,2,3\},\; M_{i}\in\{I,X,Y,Z\} \text{ for all $i\in[n]$}\},\]
    where
    $I=\begin{bmatrix}
        1 &0 \\
        0& 1
    \end{bmatrix},$
    $X=\begin{bmatrix}
        0 &&1 \\
        1 && 0\\
    \end{bmatrix},$
    $Y=\begin{bmatrix}
        0 && -i \\
        i && 0\\
    \end{bmatrix}$
    and
    $Z=\begin{bmatrix}
        1 &&0 \\
        0 && -1\\
    \end{bmatrix}$.
\end{Definition}
Using the identity $Y=iXZ$, one sees that any $g\in\mathcal{P}_{n}$ admits a representation of the form $g=i^{\lambda} \bigotimes_{i=1}^{n}X^{a_{i}}Z^{b_{i}}$, where $a,b\in\mathbb{F}_{2}^{n}$. For brevity, we write $g=i^{\lambda}X(a)Z(b)$.

To obtain stabilizer quantum codes using the stabilizer formalism one needs to find ``interesting" subgroups of $\mathcal{P}_{n}$.

\begin{Definition}
 A subgroup $\mathcal{S}\leq\mathcal{P}_{n}$ is called a stabilizer subgroup if it is abelian and $-I\notin\mathcal{S}$.
\end{Definition}
Once the stabilizer subgroup is obtained, the corresponding stabilizer quantum code can be obtained as the joint $+1$ eigenspace of elements in the subgroup.
\begin{Definition}
    Let $\mathcal{S}$ be a stabilizer subgroup of the Pauli group $\mathcal{P}_{n}$. The stabilizer code $C(\mathcal{S})$ determined by $\mathcal{S}$ is the subspace of states fixed by every element of $\mathcal{S}$, i.e., the common $+1$-eigenspace
    \[C(\mathcal{S})=\{\ket{\psi}:\; g\ket{\psi}=\ket{\psi} \text{ for all $g\in \mathcal{S}$}\}.\]
    The code is said to encode $k$ logical qubits if $\dim(C(\mathcal{S}))=2^{k}$.
\end{Definition}

\begin{Fact}
    A stabilizer group of order $2^{k}$ gives rise to a stabilizer code of dimension $2^{n-k}$.
\end{Fact}

\begin{Fact}
    Any operator $A\in M_{n}(\mathbb{C})$ on $n$ qubits can be written as
    $$A=\sum_{a,b\in \FF_{2}^{n}}\lambda_{a,b}X(a)Z(b),$$
    for suitable coefficients $\lambda_{a,b}\in \mathbb{C}$, and this decomposition of $A$ over the operators $X(a)Z(b)$ is unique.
\end{Fact}

Having obtained the stabilizer quantum code, we need to verify its error resilience, and for that we need a notion of error correction, as in the classical case.

\begin{Definition}
    Suppose an operator $E\in M_{n}(\mathbb{C})$ has the decomposition
    $$E=\sum_{a,b\in \FF_{2}^{n}}\lambda_{a,b}X(a)Z(b).$$
    Then the weight of $E$ is given by
    \[wt(E) \defeq \max\{wt_{H}(a\lor b):\lambda_{a,b}\neq 0\},\]
    where $a\lor b$ is the bitwise OR of $a, b \in \FF_{2}^{n}$.
\end{Definition}

By the Knill--Laflamme error-correction conditions~\cite{knill1997theory}, a Pauli error $E$ is detectable by a quantum code $Q$ spanned by an orthonormal basis $\{\ket{c_i}\}_{i=1}^{K}$ if and only if
\begin{equation}
  \bra{c_i} E \ket{c_j} = \alpha_{E}\delta_{ij} \quad \text{for all } i,j,
  \label{eq:KL}
\end{equation}
where $\alpha_E$ depends only on the error $E$ and not on the choice of basis states $\ket{c_i}$, $\ket{c_j}$. Consequently, the distance $d_{Q}$ of $Q$ is the smallest weight of an error $E$ for which~\eqref{eq:KL} fails.
\begin{itemize}
  \item A code is \emph{non-degenerate} if $\alpha_{E} = 0$ for every non-identity error $E \neq I_{2^{n}}$ of weight less than $d_{Q}$.
  \item A code is \emph{degenerate} if $\alpha_{E} \neq 0$ for some non-identity error $E \neq I_{2^{n}}$ of weight less than $d_{Q}$.
\end{itemize}

Finally, we record the following fact relating the distance of a quantum code to its error-correcting capability.
\begin{Fact}
    A quantum code on $n$ qubits with distance $d_{Q}$ can correct every error $E$ of weight at most $\lfloor\frac{d_{Q}-1}{2}\rfloor$.
\end{Fact}

An important family of stabilizer quantum codes that we consider in this work are the quantum CSS codes. 
\subsection{CSS codes}
The seminal works \cite{calderbankshor98} and \cite{Steane96} showed how to build certain quantum error-correcting codes---now called \emph{CSS codes}--- using any pair of classical binary linear codes $(C_1,C_2)$ with $C_2^\perp \le C_1$. The stabilizer group for the CSS code defined by $(C_{1}, C_{2})$ looks like: $<X(c_{1})Z(c_{2}):c_{1}\in C_{1}^{\perp},\; c_{2}\in C_{2}^{\perp}>.$
If $C_1$ and $C_2$ have parameters $[n,k_1,d_1]$ and $[n,k_2,d_2]$, then the CSS code $CSS(C_1, C_2)$ has parameters $[[n, k_1+k_2-n, d_{CSS}]]$ where
\[ d_{CSS} \defeq \min\{wt_H(a):a\in (C_1\setminus C_2^{\perp})\cup (C_2\setminus C_1^\perp )\} \]
is the (quantum) distance of $CSS(C_1, C_2)$.
If the parity check matrices of $C_1$ and $C_2$ are $H_1$ and $H_2$, then the condition that $C_2^\perp \le C_1$ is equivalent to to the condition $H_1 H_2^T=O$, and we will write $CSS(H_1,H_2)$ to mean $CSS(C_1,C_2)$, and may refer to $H_1$ and $H_2$ as the \emph{quantum parity-check matrices} of the CSS code.

The quantum distance problems of primary interest in this work are the CSS analogs of the classical problems $\MinDist, \MultGapDist$ and $\AddGapDist$.
Generalizing the first two is completely straightforward.
As before, let $\gamma\ge1$.

\begin{quote}
\underline{CSS Minimum Distance Decision Problem ($\CSSMinDist$)}
\label{CSSMinDist}

\textbf{Instance:} $H_X\in\mathbb{F}_{2}^{(n-k_{1})\times n}$ and $H_{Z}\in\mathbb{F}_{2}^{(n-k_{2})\times n}$ with $H_{X}H_{Z}^{T}=O$, and a non-negative integer $t$.

\textbf{Output:} YES if $d_{CSS}\leq t$. NO otherwise. 
\end{quote}

\begin{quote}
\label{MultGapCSSDist}
\underline{CSS Minimum Distance Decision Problem with Multiplicative Gap $\gamma$}

\underline{($\MultGapCSSDist_{\gamma}$)}

\textbf{Instance:} $H_X\in\mathbb{F}_{2}^{(n-k_{1})\times n}$ and $H_{Z}\in\mathbb{F}_{2}^{(n-k_{2})\times n}$ with $H_{X}H_{Z}^{T}=O$, and a non-negative integer $t$.

\textbf{Promise:} Either $d_{CSS}\leq t $ or $d_{CSS}>\gamma t$.

\textbf{Output:} YES if $d_{CSS}\leq t $. NO if $d_{CSS}>\gamma t$.
\end{quote}
There is an important subtlety in generalizing $\AddGapDist$ to the CSS setting.
We need \emph{two} parameters for the problem: $\tau>0$ as before and a new $\epsilon\in(0,1]$.
\begin{quote}
\label{AddGapCSSDist}
\underline{CSS Minimum Distance Decision Problem with Additive Gap $\tau n^\epsilon$}

\underline{($\AddGapCSSDist_{\tau,\epsilon}$)}

\textbf{Instance:} $H_X\in\mathbb{F}_{2}^{(n-k_{1})\times n}$ and $H_{Z}\in\mathbb{F}_{2}^{(n-k_{2})\times n}$ with $H_{X}H_{Z}^{T}=0$, and a non-negative integer $t$.

\textbf{Promise:} Either $d_{CSS} \leq t$ or $d_{CSS}>t+\tau n^\epsilon$.

\textbf{Output:} YES if $d_{CSS} \leq t$.  NO if $d_{CSS}>t+\tau n^\epsilon$.
\end{quote}


We conclude this review section with some facts about the (non-)degeneracy of quantum CSS codes.
\begin{Fact}\label{fact:degenerate}
A stabilizer code with distance $d_Q$ is \emph{degenerate} if its stabilizer
group contains a non-identity Pauli operator of weight less than $d_Q$, and
\emph{non-degenerate} otherwise. For a CSS code $\mathrm{CSS}(C_1,C_2)$ with
$k \ge 1$, this condition depends only on $C_1$ and $C_2$:
\[
  \mathrm{CSS}(C_1,C_2) \text{ is non-degenerate}
  \iff d_Q = \min\{d(C_1), d(C_2)\}
  \iff \min\{d(C_1^{\perp}), d(C_2^{\perp})\} \ge d_Q .
\]
Equivalently, it is degenerate if and only if
$\min\{d(C_1^{\perp}), d(C_2^{\perp})\} < d_Q$.
\end{Fact}
\subsubsection{Hypergraph product codes}
An important class of CSS codes are the 
\emph{hypergraph product codes}, where \emph{any} two classical linear codes (of \emph{any} lengths) are combined in a way that directly yields a CSS code \cite{Tillich13}.
\begin{Theorem}[\cite{Tillich13}]
    If $H_1$ and $H_2$ are the parity-check matrices for $C_1$ and $C_2$, then the \emph{hypergraph product code} $HGP(H_1,H_2)$ is the CSS code with quantum parity check matrices $H_1'$ and $H_2'$ defined as follows:
\[H_1' \defeq [H_1 \otimes I : I \otimes H_2^T], \qquad H_2' \defeq [I \otimes H_2 : H_1^T \otimes I].\]
Moreover, if $H_{1}$ and $H_{2}$ have full rank, and $C_{1}$ and $C_{2}$ are codes 
with parameters $[n_{1}, k_{1}, d_{1}]$ and $[n_{2}, k_{2}, d_{2}]$, 
respectively, then 
$\mathrm{HGP}(H_{1}, H_{2})$ is a quantum code with parameters
\[
    [[\, n_{1}n_{2} + (n_{1}-k_{1})(n_{2}-k_{2}),\; k_{1}k_{2},\; \min\{d_{1}, d_{2}\} \,]].
\]
\end{Theorem}

It is straightforward to verify that $H_1'(H_2')^T=O$, and hence, this defines a valid CSS code.

Put succinctly, when given full rank classical parity-check matrices, we have
\begin{equation}
\label{eqn:TZ}
HGP([n_1,k_1,d_1],[n_2,k_2,d_2]) =[[n_1n_2+(n_1-k_1)(n_2-k_2),k_1k_2, \min\{d_1,d_2\}]]
\end{equation}
\subsection{Graph States}
For the CWS code based on a classical code $C$ with distance $d$ and a graph $G$, the quantum distance $d_Q$ is known to satisfy the inequality $d_Q \le \min\{d, d_G\}$ where $d$ is the classical distance of $C$ and $d_G$ is the \emph{minimum graph state distance}
\[d_G \defeq \min\{wt_{H}(x\lor z):A_{G}x^{T}=z^{T}, \; (x,z)\neq O_{2n\times 1}\}.\]
Here, $A_G$ denotes the adjacency matrix of $G$, and $O_{k,l}$ represents a $k \times l$ all-zeros matrix. When the dimensions are clear from context, the subscript on $O$ may be omitted.

We now define graph states, which we touched on in the introduction.
\begin{Definition}
    For a graph $G$ on $n$ vertices specified by an adjacency matrix $A_{G}$, the graph state is defined as the stabilizer code corresponding to the stabilizer group
    \[\langle X(e_{i})Z(u_{i}):i\in[n]\rangle.\]
    where $e_{i}$ is the $i^\text{th}$ standard basis vector and $u_{i}$ is the $i^\text{th}$ row of $A_{G}$.
    Note that we refer to it as a graph ``state" rather than a graph ``code" because it is always 1-dimensional, hence, can be understood as a single quantum state $\ket{\psi_G}$.
\end{Definition}

The following are the formal definitions of the decision problems associated to graph states that we study in Theorem \ref{Graph-State-Distance}. 
\begin{quote}
\underline{Graph State Minimum Distance Decision Problem ($\GraphMinDist$)}

\textbf{Instance:} The adjacency matrix $A_{G}$ of a simple graph and a non-negative integer $t$.

\textbf{Output:} YES if $d_{G}\leq t$.  NO otherwise.
\end{quote}
\begin{quote}
\underline{Graph State Minimum Distance Decision Problem with Multiplicative Gap $\gamma$}

\underline{($\MultGapGraphDist_\gamma$)}

\textbf{Instance:} The adjacency matrix $A_{G}$ of a simple graph and a non-negative integer $t$.

\textbf{Promise:} Either $d_{G}\leq t$ or $d_{G}>\gamma t$.

\textbf{Output:} YES if $d_{G}\leq t$ and NO if $d_{G}>\gamma t$.
\end{quote}
\begin{quote}
\underline{Graph State Minimum Distance Decision Problem with Additive Gap $\tau n^\epsilon$}

\underline{($\AddGapGraphDist_{\tau,\epsilon}$)}

\textbf{Instance:} The adjacency matrix $A_{G}$ of a simple graph and a non-negative integer $t$.

\textbf{Promise:} Either $d_{G}\leq t$ or $d_{G}>t+\tau n^\epsilon$.   

\textbf{Output:} YES if $d_{G}\leq t$ and NO if $d_{G}>t+\tau n^\epsilon$.   
\end{quote}

A crucial observation is that graph states, viewed as stabilizer codes, are degenerate.
\begin{Lemma}\label{lem:graph-state-non-degeneracy}
    Let $G$ be a graph on $n$ vertices with adjacency matrix $A_G$, and let
    $\ket{\psi_G}$ denote the corresponding graph state. Then $\ket{\psi_G}$
    is a degenerate stabilizer code encoding zero logical qubits.
\end{Lemma}

\begin{proof}
    Let $d_G$ denote the minimum distance of the graph state $\ket{\psi_G}$,
    and recall that its stabilizer group is generated by the operators
    $K_i = X(e_i)Z(u_i)$ for $i \in [n]$, where $e_i \in \mathbb{F}_2^n$ is
    the $i$-th standard basis vector and $u_i \in \mathbb{F}_2^n$ is the
    $i$-th row of $A_G$. For vectors
    $(a_1, b_1), (a_2, b_2) \in \mathbb{F}_2^n \times \mathbb{F}_2^n$, define
    the symplectic inner product
    \[
        \langle (a_1, b_1), (a_2, b_2) \rangle_{\mathrm{symp}}
        :=\langle a_1, b_2 \rangle + \langle b_1, a_2 \rangle .
    \]

    Consider an error
    \[
        E = \sum_{a, b \in \mathbb{F}_2^n} \lambda_{a,b}\, X(a) Z(b)
        \qquad \text{with} \qquad
        0< \operatorname{wt}(E) \leq d_G - 1 \qquad \text{and}\qquad \lambda_{0,0} \neq 0 .
    \]
    By linearity,
    \[
        \bra{\psi_G} E \ket{\psi_G}
        = \sum_{a, b \in \mathbb{F}_2^n} \lambda_{a,b}\,
          \bra{\psi_G} X(a) Z(b) \ket{\psi_G} .
    \]

    Since $0<\operatorname{wt}(E) \leq d_G - 1$, for every
    pair $(a, b)$ with $(a,b)\neq (0,0)$ and $\lambda_{a,b} \neq 0$ there exists an $i \in [n]$
    such that
    \[
        \langle (e_i, u_i), (a, b) \rangle_{\mathrm{symp}}
        = \langle e_i, b \rangle + \langle u_i, a \rangle = 1,
    \]
    i.e., $X(a)Z(b)$ anticommutes with the stabilizer generator $K_i$.
    Using $K_i \ket{\psi_G} = \ket{\psi_G}$ and the commutation relation, we
    obtain
    \begin{align*}
        \bra{\psi_G} X(a) Z(b) \ket{\psi_G}
        &= \bra{\psi_G} X(e_i) Z(u_i)\, X(a) Z(b) \ket{\psi_G} \\
        &= (-1)^{\langle (a,b), (e_i, u_i) \rangle_{\mathrm{symp}}}
           \bra{\psi_G} X(a) Z(b)\, X(e_i) Z(u_i) \ket{\psi_G} \\
        &= -\bra{\psi_G} X(a) Z(b) \ket{\psi_G},
    \end{align*}
    which forces
    \[
        \bra{\psi_G} X(a) Z(b) \ket{\psi_G} = 0
    \]
    for all such $(a, b)$. But for the case where $(a,b)=(0,0)$ and $\lambda_{0,0}\neq 0$ we have:
    \[\bra{\psi_G} \lambda_{0,0}I \ket{\psi_G} = \lambda_{0,0}\neq 0\]
    
    Therefore
    \[
        \bra{\psi_G} E \ket{\psi_G} = \lambda_{0,0} = \alpha_E\neq 0 ,
    \]
    as claimed.
\end{proof}

\section{Proof of Theorem \ref{thm:CSSdist}}
\label{sec:thm1}

While the hardness of $\CSSMinDist$ and $\MultGapCSSDist$ follows directly from that of $\MinDist$ and $\MultGapDist$ respectively, establishing the hardness of $\AddGapCSSDist_{\tau,\epsilon}$ requires a reduction from an $\epsilon$-parameterized variant of $\AddGapDist_\tau$, which we introduce below.
\begin{quote}
\label{epsilonadditivegap}
\underline{Classical Minimum Distance Decision Problem with Additive Gap $\tau n^{\epsilon}$}

\underline{($\AddGapDist_{\tau, \epsilon}$)}

\textbf{Instance :} Binary matrix $\textbf{$H$}\in\mathbb{F}_{2}^{n-k\times n}$ and a positive integer $t$

\textbf{Promise:} Either $d\leq t$ or $d>t+\tau n^{\epsilon}$.

\textbf{Output:} YES if $d\leq t$ and NO if $d>t+\tau n^{\epsilon}$
\end{quote}

We show that the decision problem $\AddGapDist_{\tau,\epsilon}$, as defined above, is also $\NP$-hard. The proof relies on certain properties of the tensor product of codes, which we recall below.
\begin{Definition}
    For classical codes $C_{1},$ $C_{2}$ with parameters $[n_{1},k_{1},d_{1}]$ and $[n_{2},k_{2}, d_{2}]$ and generator matrices $G_{1},\; G_{2}$ respectively, the tensor product code, $C_{1}\otimes C_{2}$, is defined as the code corresponding to the generator matrix $G_{1}\otimes G_{2}$.
   
\end{Definition}

\begin{Lemma}
For every $\epsilon\in(0,1]$, there exists $\tau\in(0,1)$ for which $\emph{\AddGapDist}_{\tau, \epsilon}$ is $\mathsf{NP}$-hard.
\end{Lemma}

\begin{proof}
  We reduce from \textsc{AddGapDist}$_{\tau}$. Consider an instance $(H,t)$ for the \textsc{AddGapDist}$_{\tau}$ problem, where $\tau$ is as defined in \cite{dumer2003approx} (Theorem 32) such that \textsc{AddGapDist}$_{\tau}$ is $\NP$-hard. Let $C$ be the code of length $n$ corresponding to the parity-check matrix $H$ and let $d$ be its minimum distance.  Now, consider the tensor code $C^{'}=C\otimes \mathbb{F}_{2}^{\lceil n^{1/\epsilon-1}\rceil}$. By Fact \ref{Tensor-product}, it follows that the distance of $C^{'}$ is $d$ and its length is $N=n\cdot \lceil n^{1/\epsilon-1}\rceil$. We claim that
    $\min\{d,{N}^{\epsilon}\}=d$. Since $\min\{d,N^{\epsilon}\}\leq d$,  it suffices to show that $\min\{d,N^{\epsilon}\}\geq d$, which follows from: 
    \begin{align*}
\min\{d,N^{\epsilon}\}=&\min\{d,({n\cdot \lceil n^{1/\epsilon-1}\rceil})^{\epsilon}\}\\&\geq \min \{d,({n\cdot n^{1/\epsilon-1}})^{\epsilon}\}\\
&=\min\{d,n\}\\
&=d.
\end{align*}

 By definition $N=n\cdot \lceil n^{1/\epsilon-1}\rceil$. Note that as $\lceil n^{1/\epsilon}\rceil\leq 2n^{1/\epsilon}$, we have $n^{1/\epsilon}\leq N\leq 2n^{1/\epsilon}.$   
If $(H,t)$ is a YES instance for $\textsc{AddGapDist}_{\tau}$, then $d\leq t$, and as above it follows that $\min\{d,N^{\epsilon}\}\leq t$. If $(H,t)$ is a NO instance for $\textsc{AddGapDist}_{\tau}$, then $\min\{d,N^{\epsilon}\}=d>t+\tau n\geq t+\tau^{'} N^{\epsilon},$ where $\tau'=\tau/{2^{\epsilon}}.$
This establishes the lemma.
\end{proof}

We now provide a proof of Theorem \setcounter{hereCSS}{\value{Theorem}}
\setcounter{Theorem}{\value{savedCSS}}\addtocounter{Theorem}{-1}
\begin{restatable}{Theorem}{cssdistthm}
\label{thm:CSSdist}
Each of the problems listed below is $\NP$-hard under Karp reductions:
\begin{itemize}
\item $\emph{\CSSMinDist}$
\item $\emph{\MultGapCSSDist}_\gamma$ for every $\gamma \ge 1$
\item $\emph{\AddGapCSSDist}_{\tau,\epsilon}$, for each $0 < \epsilon \le \frac{1}{2}$ and some $\tau>0$ (depending on $\epsilon$)
\end{itemize}
\end{restatable}
\setcounter{Theorem}{\value{hereCSS}}%
\begin{proof}[Proof of Theorem \ref{thm:CSSdist}]
To see that $\textsc{CSSMinDist}$ is $\NP$-hard, consider an input instance $(H_1,t)$ of the $\textsc{MinDist}$ problem, where $H_1$ is the parity-check matrix of a classical $[n,k,d]$ code $C$.
We may assume $H_1$ is full rank (if it is not, then we may, in polynomial time, replace it with a new, smaller parity-check matrix that is).
Now let $H_2$ be the (full rank) parity-check matrix for the classical repetition code of length $n$, which has parameters $[n,1,n]$.
We reduce to the hypergraph product code $HGP(H_1,H_2)$.
By Equation \ref{eqn:TZ}, this hypergraph product code has parameters $[[n^2+(n-k)(n-1),k,d]]$.
This implies the hardness of $\CSSMinDist$. 

The reduction just shown builds a quantum code whose minimum distance equals the distance of the original classical code.
Hence, even finding a constant multiplicative approximation to $\CSSMinDist$ is $\NP$-hard.
That is, $\MultGapCSSDist$ is $\NP$-hard.

We now reduce the  $\AddGapDist_{\tau, \epsilon}$ problem to the $\AddGapCSSDist_{\tau^{'},\epsilon^{'}}$ problem where $\tau'=\tau/2^{\frac{\epsilon}{1+\epsilon}}$ and $\epsilon'=\frac{\epsilon}{1+\epsilon}$, hence showing the hardness of the latter problem. For a given $\epsilon\in(0,1]$, consider the repetition code of length $n^{\epsilon}$.
Let $(H,t)$ be the input instance for $\AddGapDist_{\tau, \epsilon}$ and let $C$ be the code corresponding to parity-check matrix $H$ having length $n$ and distance $d$.
The hypergraph product code obtained from the code $C$ and the repetition code will be of length $n\cdot n^{\epsilon}+(n-k)\cdot(n^{\epsilon}-1)=n^{'}$ and distance $d^{'}=\min\{d,n^{\epsilon}\}$. 

We claim that the above conversion of a classical code $C$ to a hypergraph product code is a deterministic reduction from the $\AddGapDist_{\tau, \epsilon}$ problem to the $\AddGapCSSDist_{\tau^{'},\epsilon^{'}}$ problem. 

If $(H,t)$ is a YES instance of $\AddGapDist_{\tau, \epsilon}$, then $\min\{d,n^{\epsilon}\}=d'\leq t$.
This implies that the corresponding hypergraph product code has distance $d'\leq t$.
If $(H,t)$ is a NO instance, then $d'>t+\tau n^{\epsilon}\geq t+\tau' (n')^{\frac{\epsilon}{1+\epsilon}},$ where $\tau'=\tau/{2^{\frac{\epsilon}{1+\epsilon}}}.$
This follows from the observation $n^{1+\epsilon}\leq n'\leq 2n^{1+\epsilon}.$
Clearly, the gap between the YES instances and NO instances of $\AddGapCSSDist_{\tau,\epsilon}$ is maximized for $\epsilon=1$, which is the case when the NO instance is $d'\geq t+\tau^{'} n^{1/2}$.
\end{proof}

\section{Hardness for distances of graph states}
\subsection{Proof of Theorem \ref{Graph-State-Distance}}
\label{ss:GSD}

We now prove the hardness of computing the minimum graph state distance.
We reduce \textsc{MinDist} to \textsc{GraphDist} via  intermediate variants of  \textsc{MinDistDualDist}. Similarly, to prove hardness for the gap version of $\GraphMinDist$ we will need the corresponding gap versions \textsc{MultGapMinDistDualDist} and \textsc{AddGapDistDualDist} as defined next.
\noindent 
\vspace{2pt}
\begin{quote}\label{MinDualDist}
\underline{Minimum Distance Dual Distance (\textsc{MinDistDualDist})} 

 \textbf{Instance: }Binary matrix $H\in\mathbb{F}_{2}^{n-k\times n}$ and positive integer $t$.
 
 \textbf{Output:} If $\min\{d(C),\; d(C^{\perp})\}\leq t$ then YES else NO.

\end{quote}
\begin{quote}\label{MultGapMinDualDist}
\underline{Minimum Distance Dual Distance with Multiplicative Gap $\gamma$}

\underline{(\textsc{MultGapMinDistDualDist}$_{\gamma}$)} 

 \textbf{Instance: }Binary matrix $H\in\mathbb{F}_{2}^{n-k\times n}$ and positive integer $t$.
 
 \textbf{Output:} If $\min\{d(C),\; d(C^{\perp})\}\leq t$ then YES and NO if $>\gamma t$.

\end{quote}

\begin{quote}\label{AddGapMinDualDist}
\underline{Minimum Distance Dual Distance with Additive Gap $\tau n^{1/3}$}

\underline{(\textsc{AddGapMinDistDualDist})} 

 \textbf{Instance: }Binary matrix $H\in\mathbb{F}_{2}^{n-k\times n}$ and positive integer $t$.
 
 \textbf{Output:} If $\min\{d(C),\; d(C^{\perp})\}\leq t$ then YES and NO if $\min\{d(C),\; d(C^{\perp})\}>t+\tau n^{1/3}$.

\end{quote}




\begin{Lemma}
\label{lem:dualdist}
    $\emph{\MinDistDualDist}$ is $\mathsf{NP}$-hard.
    Furthermore, for every constant $\gamma\geq 1$, computing a \newline $\emph{\MultGapDistDualDist}$ approximation is $\mathsf{NP}$-hard while there exists a $\tau\in(0,1)$ such that finding an additive approximation with error $\tau n^{1/3}$ is $\mathsf{NP}$-hard. 
\end{Lemma}

\begin{proof}

For a given parity-check matrix $H$, let $C$ be the associated code to it with length $n$ and minimum distance denoted by $d(C)$. Now, consider a code $C^{'}$, that comes from a family of codes $\{C^{'}_{N}\}$ where $C^{'}_{N}$ is a binary code of length $N$ such that $\min\{d(C^{'}),\; d({C^{'}}^{\perp})\}\geq \sqrt{N}$.

One such family of codes was introduced in \cite{Hai2023infinte} (Theorem 2.9, Theorem 2.15). In the appendix, we provide another, and arguably simplified, construction of such an infinite family of codes with $\min\{d(C'),\; d({C^{'}}^{\perp})\}= \Omega(\sqrt{N}).$

While \cite{Hai2023infinte} does not provide any analysis of the time required to build their codes' parity-check matrices, one can check that this can be done efficiently.
(To do so, one must use the fact that the finite field $\mathbb{F}_{2^{m}}$ of order $2^{m}$ with $m={O(\log {n})}$ can be constructed in time $O(poly(n))$.)
We also use the observation that for every integer $n\geq 5$ there exists $m\equiv 5 \mod{6}$ for which $2^{m}\leq n\leq 2^{m+6}$. The aforementioned observation along with the fact \ref{code-elongation} allows us to increase the length of the code while preserving the distance.

In the Appendix, we provide another construction of such an infinite family of classical error-correcting codes in Lemma \ref{Lemma:sqrtcodedualcode}.

We consider the tensor code $\tilde{C}:=C^{\perp}\otimes C^{'}$ and its dual $\tilde{C}^{\perp}:=C\otimes \FF_{2}^{N}+\FF_{2}^{n}\otimes {C^{'}}^{\perp}$

Using Fact \ref{Tensor-product} about the tensor-product of codes, we have:
\[\min\{d(\tilde{C}),d(\tilde{C}^{\perp})\}=\min\{d(C^\perp)d(C'),d(C),d({C'}^{\perp})\}\leq d(C).\]

  Given that the length of code $C$ is $n$, and  $d(C'),\; d({C'}^{\perp})$  are both greater than or equal to $\sqrt{N}$, we get 
    \[\min\{d(\tilde{C}),d(\tilde{C}^{\perp})\}\geq \min\{d(C^\perp)\cdot n,d(C),n\}\geq d(C),\] for $N\geq n^{2}$.  Hence, $\min\{d(\tilde{C}),d(\tilde{C}^{\perp})\}=d(C)$. For the error in additive approximation, as we are embedding code of length $n$ in a space of length $n^{3}$, hence we get the additive approximation term with cubic error. 
\end{proof}

We now reduce this intermediate problem to \textsc{GraphMinDist} and complete the proof of Theorem \ref{Graph-State-Distance}.

\setcounter{hereGraph}{\value{Theorem}}
\setcounter{Theorem}{\value{savedGraph}}\addtocounter{Theorem}{-1}
\begin{restatable}{Theorem}{graphdistthm}
\label{Graph-State-Distance}
Each of the following is $\NP$-hard under Karp reduction:
\begin{itemize}
    \item $\emph{\GraphMinDist}$
    \item $\emph{\MultGapGraphDist}_\gamma$ for any $\gamma \ge 1$
    \item $\emph{\AddGapGraphDist}_{\tau,\epsilon}$ for each $0 < \epsilon \le \frac{1}{3}$ and some $\tau>0$ (depending on $\epsilon$).
\end{itemize}
\end{restatable}
\setcounter{Theorem}{\value{hereGraph}}

\begin{proof}[Proof of Theorem \ref{Graph-State-Distance}]

    For an input instance $(H,t)$ of \textsc{MinDistDualDist} consider the systematic form representation of $H$ i.e., $H=[I:M]$. Now, consider the following symmetric matrix $A_{M}$ corresponding to $H$, i.e.:
    \[A_{M}=\begin{bmatrix}
O_{n-k\times n-k} & M\\
       M^{T} & O_{k\times n-k} \\
\end{bmatrix}.\]

It can be verified that as $A_M$ is a symmetric matrix hence it can be treated as the adjacency matrix of a graph (infact a simple graph as the diagonal entries are all zero.)

As mentioned in the definition of minimum graph-state distance, the minimum graph state distance of $A_{M}$  is obtained by finding $x:=(a_{1},...,a_{n}|b_{1},...,b_{n})\neq 0$ with minimum $wt_{H}(a\lor b)$ such that

\begin{equation}[I:A_{M}]x^{T}=O\end{equation}

Now, the system of equations given in the above equation can be broken down into the following system of equations.

  \begin{equation}
  \bigg[I:M\bigg]\begin{bmatrix}
        a_{1}\\
        a_{2}\\
        \vdots\\
        a_{n-k}\\
        b_{n-k+1}\\
        \vdots\\
        b_{n}\\
    \end{bmatrix}
    =O_{n-k\times 1}
    \end{equation}

\begin{equation}\bigg[M^{T}:I\bigg]\begin{bmatrix}
        b_{1}\\
        \vdots\\
        b_{n-k}\\
        a_{n-k+1}\\
        \vdots\\
        a_{n}\\
    \end{bmatrix}
    =O_{k\times 1}
\end{equation}

Let $S_{1},\; S_{2}$ be the set of variables that correspond to the above two system of equations respectively. As the sets $S_{1}$ and $S_{2}$ are disjoint, the minimum graph state distance for $A_{M}$, is obtained by assigning one set of variables the value $0$, and for the remaining equation, we can find the non-zero solution with the least Hamming weight. This implies that the  minimum graph-state distance corresponding to the graph whose adjacency matrix is $A_{M}$ is:
\[d_{G}=\min\{d(C),d(C^{\perp})\},\]
and completes the proof of the hardness of \GraphMinDist. Moreover, the above reduction also shows hardness for $\MultGapGraphDist$  and \AddGapGraphDist. 
\end{proof}

As a corollary to Lemma \ref{lem:dualdist}, we obtain that computing the distance of classical linear codes having rate $1/2$ is $\NP$-complete.
\begin{Corollary}
    It is $\mathsf{NP}$-hard under Karp reductions to compute the distance of constant rate linear codes with rate $\in (0,1/2]$. Furthermore, it is $\mathsf{NP}$-hard under Karp reduction to compute both constant factor approximations and $n^{1/3}$ additive approximations to the distance of constant rate linear codes with rate $\in (0,1/2]$.  
\end{Corollary}
\begin{proof}
Consider a code $C$ with parameters $[n,k,d]$ and parity-check matrix $[I_{n-k}:M_{n-k\times k}]$.
Now consider the matrix:
\[\begin{bmatrix}
    I_{k} &O_{k \times n-k} &:& O_{n-k \times k} & M_{n-k\times k}\\
    O_{n-k \times k} & I_{n-k}&: &M_{k\times n-k}^{T} &O_{k \times n-k}\\
\end{bmatrix}.\]
The above matrix can be treated as a parity-check matrix corresponding to a rate $1/2$ code. Computing the distance of the above code is computing a non-zero vector $(x,y)\in\mathbb{F}_{2}^{n}\oplus\mathbb{F}_{2}^{n}$ of least Hamming weight for which:
\[\begin{bmatrix}
    I_{k} & O_{k \times n-k} &:& O_{n-k \times k} & M_{n-k\times k}\\
    O_{n-k \times k} & I_{n-k}&: &M_{k\times n-k}^{T} & O_{k \times n-k}\\
\end{bmatrix}\begin{bmatrix}
    x^T\\
    y^T
\end{bmatrix}=O_{n\times 1}.\]
As in the proof of Theorem \ref{Graph-State-Distance}, one sees the distance of this code is $\min\{d(C),d(C^{\perp})\}$. To prove the hardness of codes with having rate $\in (0,1/2)$, we take a code of rate equal to 1/2 and pad it with appropriate number of zeros. By Fact \ref{code-elongation}, this preserves the distance but decreases the rate.
\end{proof}

\subsection{Proof of Theorem \ref{X-Graphstatedistance}}
\label{ss:XGSD}

\begin{claim}\label{claim:self-orthogonal}
    Let $C$ be a $[n,k,d]$ be a binary linear code. Consider the map $\pi:\mathbb{F}_{2}^{n}\rightarrow(\mathbb{F}_{2}^{n})^{2}$ defined as:
    \[\pi(a)=(a,a).\]
    Then the image of $C$, defined as:
    \[\pi(C)=\{\pi(c):c\in C\},\]
    is a self-orthogonal linear code with parameters $[2n,k,2d].$
\end{claim}
\begin{proof}
    We first prove that the map $\pi(C)$ is a linear code of length $2n$.
    For any two vectors $\pi(c),\pi(c^{'})\in\pi(C)$ and $\lambda\; \mu\in\mathbb{F}_{2}$ we have that:
    \[\lambda \pi(c)+\mu\pi(c^{'})=\lambda(c,c)+\mu(c^{'},c^{'})=(\lambda c+\mu c^{'},\lambda c+\mu c^{'}).\]

    As $C$ is a linear subspace hence $\lambda c+\mu c^{'}\in C$, hence $(\lambda c +\mu c^{'},\lambda c+\mu c^{'})\in \pi(C)$. 

    We now prove that $dim(\pi(C))=dim(C)=k$.
    This follows from the observation that the map $\pi$ is an injective map. Indeed,
    $\pi(c)=(O_{1\times n},O_{1\times n})$ implies $(c,c)=O_{1\times 2n}$, \emph{i.e.}\ $c=O_{1\times n}$.

    The distance of the linear code $\pi(C)$ is calculated by observing that 
    \[
    \begin{aligned}
        \min_{\pi(c)\in\pi(C):\pi(c)\neq (O,O)}wt_{H}(\pi(c))&=\min_{c\in C:c\neq O}\pi(c)\\
        &=\min_{c\in C:c\neq O}2wt_{H}(c) =2\min_{c\in C:c\neq O}wt_{H}(c)=2d_{H}(C).
    \end{aligned}
    \]

    Now, to prove that $\pi(C)$ is self-orthogonal we need to show that any pair of codewords $\pi(c),\; \pi(c^{'})\in\pi(C)$ is orthogonal to each other.

    Consider the inner-product:
    \[\langle\pi(c),\pi(c^{'})\rangle=\langle(c,c),(c^{'},c^{'})\rangle=\langle c,c'\rangle+\langle c,c'\rangle=0,\]
    
    where the last equality is due to the fact that we are in the field, $\mathbb{F}_{2}$ which has characteristic 2.
    This proves the self-orthogonality of $\pi(C)$.
\end{proof}

Now we describe our reduction from $\MinDist$ to $\emph{\GraphMinDist}^X$.
\setcounter{hereXGraph}{\value{Theorem}}
\setcounter{Theorem}{\value{savedXGraph}}\addtocounter{Theorem}{-1}
\begin{restatable}{Theorem}{xgraphdistthm}
\label{X-Graphstatedistance}
Each of the following is $\NP$-hard under Karp reduction:
\begin{itemize}
    \item $\emph{\GraphMinDist}^X$
    \item $\emph{\MultGapGraphDist}_\gamma^X$ for any $\gamma \ge 1$
    \item $\emph{\AddGapGraphDist}_{\tau,\epsilon}^X$ for each $0 < \epsilon \le \frac{1}{3}$ and some $\tau>0$ (depending on $\epsilon$)
\end{itemize}
\end{restatable}
\setcounter{Theorem}{\value{hereXGraph}}%
\begin{proof}
Let $(H,t)$ be the instance for $\MinDist$ and let $C$ be the corresponding linear code with parameters $[n,k,d]$.
By claim \ref{claim:self-orthogonal}$, \pi(C)\in(\mathbb{F}_{2}^{n})^{2}$ is a self-orthogonal linear code with parameters $[2n,k,2d]$.
The systematic form of the parity-check matrix for $\pi(C)$: \[[I_{2n-k\times 2n-k}:M_{2n-k\times k}].\] along with its generator matrix:
 \[[M^{T}_{k\times 2n-k}:I_{k}],\]
 can be obtained in polynomial time from $H$.

 Now, consider the matrix:
 \[A_{M}=\begin{bmatrix}
     I_{k} & M_{k\times 2n-k}\\
     M^{T}_{2n-k\times k} & I_{2n-k}\\
 \end{bmatrix}.\]
 As $A_{M}$ is symmetric, the matrix $A_{M}':=\begin{bmatrix}O_{2n\times 2n} & A_{M}\\ A_{M} & O_{2n \times 2n}\end{bmatrix}$ is symmetric and hence can be treated as a (simple) graph.
Now, consider vectors $(x,z)\in\mathbb{F}_{2}^{4n}\oplus\mathbb{F}_{2}^{4n}$ with $z=O_{1\times 4n}$ such that
\begin{align*}
[I:A_{M}^{'}](z,x)^{T}&=O_{4n\times 1}\\ \intertext{or equivalently,}
A^{'}_{M}{(x_{1},x_{2})}^{T}&= \begin{bmatrix}
    A_{M}x_{2}^{T}\\
    A_{M}x_{1}^{T}
\end{bmatrix}=O_{4n\times 1}.\intertext{where $x=(x_{1},x_{2})\in\mathbb{F}_{2}^{2n}\oplus\mathbb{F}_{2}^{2n}$.}\\ \intertext{As we are finding the vector $x$ with least Hamming weight, we can assume that $x_{2}=O$. Hence, solving the above equation is equivalent to solving the following system of equations:}
[I_{k}:M_{2n-k\times k}]x_{1}^{T}&=O_{2n\times1}\\ \intertext{and,}
[M_{2n-k\times k}^{T}:I_{k}]x_{1}^{T}&=O_{2n\times 1}
\end{align*}
Thus, we are finding those vectors $x_{1}\in\mathbb{F}_{2}^{2n}$ for which both
$Hx_{1}^{T}=O_{(2n-k)\times 1}$ and
$Gx_{1}^{T}=O_{k\times 1}$.  In other words, we are finding vectors $x_{1}\in\mathbb{F}_{2}^{2n}$ for which
$x_{1}\in \pi(C)\cap \pi(C)^{\perp}$.
Since $\pi(C)$ is self-orthogonal, it follows that $x_{1}\in \pi(C)$.

This implies that we want to find a non-zero vector $x_{1}\in\mathbb{F}_{2}^{2n}$ with the least Hamming weight with $x_{1}\in \pi(C)$. This is precisely the minimum Hamming distance of the classical code $\pi(C)$ which is $2d_{H}(C)$. This proves the hardness of $\GraphMinDist^{X}$ as well as $\emph{\MultGapGraphDist}_\gamma^X$ under Karp reduction for $\gamma>1$. Further, as the length of $\pi(C)$ is $2n$, $\emph{\AddGapGraphDist}_{\tau,\epsilon}^X$ is NP-hard for some $\tau\in(0,1)$.
\end{proof}

\section{Discussion and outlook}
\label{sec:outlook}

\subsection{Controlling \texorpdfstring{$\epsilon$}{e} in \texorpdfstring{AddGapHGPDist$_{\tau,\epsilon}$}{AddGapHGPDist}}
\label{ss:control}

Let $\AddGapHGPDist_{\tau,\epsilon}$ denote the version of the problem $\AddGapCSSDist_{\tau,\epsilon}$ where the input CSS code is a hypergraph product code.

\begin{Proposition}\label{Prop:square-root barrier}
  For any fixed $\epsilon\in(1/2,1]$ and $\tau>0$, there does not exist a Karp reduction from $\emph{\MinDist}$ to $\emph{\AddGapHGPDist}_{\tau,\epsilon}$ unless $\mathsf{P}=\mathsf{NP}$.
\end{Proposition}

\begin{proof}
Suppose otherwise.
Then there is a deterministic polynomial-time algorithm that takes an input instance $(H,t)$ of $\MinDist$ (where $H$ is a binary parity-check matrix and $t$ is a non-negative integer) to an instance $(H_{1},H_{2},t')$ of $\AddGapHGPDist$ in a way that preserves the respective decision problems' output.
     Say, $H_1$ and $H_2$ are full-rank and corresponding code parameters are $[n_{1},k_{1},d_{1}]$ and $[n_{2},k_{2},d_{2}]$. Recall that the hypergraph product codes corresponding to parity-check matrices $H_{1}$ and $H_{2}$ with full rank and corresponding code parameters $[n_{1},k_{1},d_{1}]$ and $[n_{2},k_{2},d_{2}]$ have parameters $[[n_{1}n_{2}+(n_{1}-k_{1})(n_{2}-k_{2}),k_{1}k_{2},\min\{d_{1},\; d_{2}\}]].$ Due to the symmetry of the construction, we can assume that $d_{1}=\min\{d_{1},\; d_{2}\}$.

    We reiterate a subtle point:  $n_{1}$ and $n_{2}$ are functions of $n$ and hence really should be expressed as $n_{1}(n)$ and $n_{2}(n)$, but we will suppress this dependency in order to avoid overloading our notation.
    Now consider the following cases:
     \begin{enumerate}
         \item For all $(H_{1},H_{2},t')$, $\min\{n_{1},n_{2}\}=O(1)$:
         If we assume that $n_{1}=\min\{n_{1}, n_{2}\}=O(1)$ then one can just perform an efficient brute-force search to decide if there is a non-zero vector  $x$ with the properties needed. While if $n_{2}\leq n_{1},$ then using the fact that $\min\{d_{1},d_{2}\}\leq \min\{n_{1},n_{2}\}=O(1)$, one can check if there is a non-zero vector $x$ with $|x|\leq n_{2}=O(1)$ such that $H_{1}x^T=0.$ This would show that $\MinDist$ admits a poly-time algorithm, which contradicts $P\neq NP.$
            
         \item $n_{1}, n_{2}=\omega(1)$: 
         Consider a NO instance $(H,t)$ of \textsc{MinDist}. By the definition of Karp reduction, it follows  $(H_{1},H_{2},t^{'})$ is a NO instance of \textsc{AddGapHGPDist}$_{\tau,\epsilon}$.
         under the assumption $d_{1}=\min\{d_{1}\; d_{2}\}$,this gives the following set of inequalities:
         \[             d_{1}>t^{'}+\tau (n_{1}n_{2}+(n_{1}-k_{1})(n_{2}-k_{2}))^{\epsilon}\geq t'+\tau (n_1 n_2)^{\epsilon}.\] 
Assume $n_1\leq n_2$ and $d_{1}=\min\{d_{1},\; d_{2}\}$. The other case can be treated similarly. 
Then $d_{1}\geq t'+\tau {n_1}^{2\epsilon}$ and so $d_{1}/n_{1}\geq t'/n_{1}+n_{1}^{2\epsilon-1}$.
Since $d_{1}\leq n_1$, one obtains a contradiction for large enough $n_{1}$. 
\end{enumerate}
 \end{proof}

There are a number of very intriguing questions that Proposition \ref{Prop:square-root barrier} stimulates.
Do there exist an $1/2<\epsilon \le 1$ and a $0< \tau < 1$ for which $\textsc{AddGapHGPDist}_{\tau,\epsilon}$ is in $\mathsf{coNP}$?
Or $\mathsf{P}$?
Or maybe $\mathsf{BQP}$?
We discuss the most intriguing question in the next subsection.

\subsection{A new square-root barrier?}
On one hand, our technique for proving the hardness of $\AddGapCSSDist$ used hypergraph product codes, and resulted in an additive approximation with a square-root error term.
Moreover, Proposition \ref{Prop:square-root barrier} shows that there is, in some sense, no way to improve on this square-root term using hypergraph product codes.
On the other hand, the proof of \cite{kapshikar2023hardness} was based on rather different methods that exploited graphs with certain extremal properties, and yet arrived at the same kind of square-root error term.
While there is no analog of Proposition \ref{Prop:square-root barrier} found in \cite{kapshikar2023hardness}, it appears that the graphs they use in their reduction are essentially optimal \cite{furedi2013history}.
It is interesting that the same square-root term in the approximation appears in both reductions, and one inevitably wonders if there is a new kind of square-root barrier at play.

\subsection{Can we get a linear approximation for minimum graph state distance and linear codes with rate 1/2?}
Recall that we were able to prove hardness for approximating the minimum graph state distance and the hardness of approximating the distance of linear codes with rate equal to 1/2 with a cube root additive term. Can we improve it? Following our proof strategy, a natural way to get past the cube root barrier and get an improvement, say a square root approximation, would be to find an infinite family of efficiently computable codes $\{C_{m}\}$ for which $\min\{d(C_{m}),\; d(C_{m}^{\perp})\}=\Omega(n_{m}),$ where $n_{m}$ is the length of $C_{m}$. Whether such an infinite family of codes exists is a possible open direction. The possibility of achieving the hardness results without the need of such a family of codes is also a direction to pursue.

\subsection{Quantum ``nearest codeword problem''}
In the classical setting, the nearest codeword problem (NCP) has played a useful role in understanding the complexity of distances of codes.
This is because NCP is essentially a non-homogeneous version of the minimum distance problem.
So it is natural to wonder about quantum analogs of the question.
One subtlety is that the codewords of a quantum code are, strictly speaking, \emph{not} vectors over $\FF_2$, but rather vectors in some subspace of the Hilbert space $(\CC^2)^{\otimes n}$.
(We have avoided discussing this point much in this paper, but see Section \ref{Preliminaries}.)
From this perspective, the quantum nearest codeword problem should be a question about computing the orthogonal projections of states onto the codespace.
There are other, discrete variants one can imagine that happen inside the symplectic $\FF_2$ vector space that contains the stabilizers of a code.
These questions would not be about codewords \emph{per se}, but about Pauli error operators, e.g.\ given a Pauli error, what is the closest logical error (either in Hamming weight or symplectic weight)?

\subsection{Rate trade-offs}
It is natural to ask about the trade-off between $k$ and $n$ for quantum
codes for which the hardness of distance has been established. Since the
quantum codes considered in this paper are CSS codes, we compare their
rates with those of the CSS codes studied in~\cite{kapshikar2023hardness}.
Recall that if $H_1$ and $H_2$ have full rank and the corresponding codes
$C_1$ and $C_2$ have parameters $[n_1,k_1,d_1]$ and $[n_2,k_2,d_2]$,
respectively, then $\mathrm{HGP}(H_1,H_2)$ is a quantum code with parameters
\[
  [[\,n_1 n_2+(n_1-k_1)(n_2-k_2),\; k_1 k_2,\; \min\{d_1,d_2\}\,]].
\]
In the reduction used in our proof of Theorem~\ref{thm:CSSdist}, we have
$n_1=n_2=n$ and $k_2=1$. The $k$-to-$n$ ratio of the resulting quantum
code is therefore
\[
  \frac{k_1}{n^{2}+(n-1)(n-k_1)} \;\leq\; \frac{1}{n},
\]
where the last inequality uses the fact that $k_1\leq n_1=n$.

To establish hardness for quantum codes with a constant $k$-to-$n$ ratio,
we replicate the reduction of Theorem~\ref{thm:CSSdist} for rate-$1/2$
classical codes. Consider an instance of a rate-$1/2$ classical code of
length $n$ with distance $d$ and a full-rank parity-check matrix $H$, and
form the hypergraph product code $\mathrm{HGP}(H,H)$. Its parameters are
\[
  [[\,n^{2}+(n/2)(n/2),\; (n/2)(n/2),\; d\,]]
  = [[\,5n^{2}/4,\; n^{2}/4,\; d\,]].
\]
Hence the $k/n$ ratio of the hypergraph product codes obtained via this
reduction is $1/5$. In contrast, the CSS codes derived from the CWS
construction of~\cite{kapshikar2023hardness} using a rate-$1/2$ classical
code with parameters $[n,k,d]$ have parameters $[[4n,2k,2d]]$, and hence
achieve a rate of $2k/4n = 1/4$. An interesting open question is whether computational hardness of
distance can be established for codes with a larger $k/n$ ratio. Note that the rate of the CSS codes
of~\cite{kapshikar2023hardness} is bounded as
$\frac{2k}{4n}=\frac{k}{2n}\leq \frac{1}{2}$. In comparison, if hardness
can be shown for classical codes of rate $R$, then the quantum CSS
codes derived from the HGP construction in the current work achieve a
$k/n$ ratio of
\[
  \frac{R^{2}}{1+(1-R)^{2}},
\]
which can be made arbitrarily close to $1$ by taking $R$ close to $1$.

\subsection{Can the hardness of computing the minimum distance be established for 
degenerate quantum codes?}
The reductions presented in this paper and in \cite{kapshikar2023hardness} 
apply only to non-degenerate quantum codes. The principal obstacle to 
extending either approach to the degenerate setting is the absence of a 
closed-form expression relating the distance of the underlying classical 
codes to the distance of the resulting quantum codes.
Practically, most quantum error correction codes implemented on actual hardware are degenerate, so
it would be very interesting if one could establish hardness results for distances in this setting.
 
 \section{Acknowledgments} 
 We thank Xuandi Ren for helpful correspondence, and the anonymous reviewers for valuable feedback.

\appendix
\section{Appendix}
\subsection{Useful facts and results}
\begin{Fact}\label{Tensor-product}
The tensor product of classical codes, $C_{1}$ and $C_{2}$ with parameters $[n_{1},k_{1},d_{1}]$ and $[n_{2},k_{2},d_{2}]$ (respectively) is a classical code with parameters $[n_{1}n_{2},k_{1}k_{2},d_{1}d_{2}]$. Moreover, $(C_{1}\otimes C_{2})^{\perp}=C_{1}^{\perp}\otimes \mathbb{F}_{2}^{n_{2}}+\mathbb{F}_{2}^{n_{1}}\otimes C_{2}^{\perp}$ and has parameters $[n_{1}n_{2},n_{1}n_{2}-k_{1}k_{2},\min\{d_{1},d_{2}\}]$  (See also Lemma 3.3 in \cite{Szegedy2023tensordual}). 
\end{Fact}
\begin{Fact}\label{code-elongation}
    For a linear code $C$ of parameters $[n,k,d]$ and parity-check matrix $H$, the code corresponding to the parity-check matrix \[H^{'}_{(n-k)\times n^{'}}=H\oplus I_{n^{'}-n}:=\begin{bmatrix}
        H_{(n-k)\times n} & O_{(n-k)\times (n^{'}-n)}\\
        O_{(n^{'}-n)\times n}& I_{(n^{'}-n)\times (n^{'}-n)}
    \end{bmatrix},\]
    has parameters $[n^{'},k,d]$.
\end{Fact}
\begin{Lemma}[\cite{Hai2023infinte}]
\label{lem:dual}
    Let $m\geq 5$ be an odd integer with $m\equiv 5\mod{6}$. For every such $m$, there exists binary classical code $C_{m}$ with parameters $[2^{m}-1,\frac{2^{m+1}-1}{3},d\geq 2^{\frac{m-1}{2}}+1]$. Moreover, $C_{m}^{\perp}$ has parameters $[2^{m}-1,\frac{2^{m}-2}{3},d^{\perp}\geq 2^{\frac{m-1}{2}}+4]$. 
\end{Lemma}
\subsection{An alternative infinite family of codes}
In this subsection, we provide an infinite family of binary linear codes $\{C_{k}\}_{k\in\mathbb{N}}$ with $\min\{d(C_{k}),\; d(C_{k}^{\perp})\}=\Omega(\sqrt{n_{k}}),$ where $n_{k}$ is the length of the classical code $C_{k}$ and which can be constructed in time $poly(n_{k})$.

To do this we consider parity-check matrices of the form $[I_{n}:A_{G}]$ where $G$ is an adjacency matrix of the graph $G$ on $n$ vertices. A nice property of such codes is that the generator matrices of these codes will be of the form $[A_{G}^{T}:I_{n}]=[A_{G}:I_{n}]$, using the fact that the adjacency matrices are symmetric matrices. Hence if $d$ is the distance of the code corresponding to a parity-check matrix $[I_{n}:A_{G}]$, its dual will also have distance $d$. Hence, we only focus on the distance of the codes obtained from parity-check matrices $[I_{n}:A_{G}]$. A similar result for matrices of the form $[I:A_{G}]$ follows from Lemma~2 of~\cite{kapshikar2023hardness}; we include a self-contained proof here for completeness. We further note that the proof of Lemma~\ref{Lemma:sqrtcodedualcode} extends naturally to codes obtained from $K_{2,t}$-free graphs, yielding a lower bound of $\underset{v \in V}{\min}\; \deg(v)/t$.
Formally, we prove the following result:

\begin{Lemma}\label{Lemma:sqrtcodedualcode}
    For every positive integer $n\geq 7$, there exists a $K_{2,2}-$free graph $G$ on $n$ vertices, with an adjacency matrix $A_{G}$, such that the code given by the parity-check matrix $[I_{n}:A_{G}]$ has distance $\Omega(\sqrt{n})$.
  \end{Lemma}

We prove the above Lemma by first showing that if a graph $G$ is $K_{2,2}-$free and has ``large minimum degree'', then the classical error-correcting code corresponding to $[I_{n}:A_{G}]$ will have large distance.

\begin{Lemma}\label{lem:graph-based-codes}
    Let $G$ be a graph on $n$ vertices and let $\delta\stackrel{def}{=}\underset{v\in V}{\min}\; deg(v)$. If $G$ is $C_{4}$ (or likewise $K_{2,2}$) free, then the classical binary linear code, say $C$, with the parity-check matrix $[I:A_{G}]$ and its Euclidean dual, say $C^{\perp}$, have parameters $[2n,n,d\geq \delta/2]$.
\end{Lemma}
\begin{proof}
    Let $x=(a_{1},...,a_{n},b_{1},...,b_{n})\in\FF_{2}^{2n}$ be a non-zero codeword of $C$ with Hamming weight $d$ such that:
    \begin{align}\label{eq:distance_equality}
    [I_{n\times n}:A_{G}]x^{T}&=O_{n\times 1},\\
    \intertext{or equivalently,}
    A_{G}b^{T}&=a^{T}.
    \end{align}
    Now let us assume that $|x|<\delta/2$. This implies, $|a|,\; |b|<\delta/2$.

    Let $supp(b)\defeq \{i\in[n]: b_{i}\neq 0\}=\{i_{1},...,i_{l}\}$. This reformulates equation \ref{eq:distance_equality} to:
    $$\sum_{k\in [l]}A_{i_{k}}=a,$$
    where $A_{j}$ denotes the $j-$th column of $A$ treated as a binary vector of dimensions $n\times 1$.

Now, let us consider the left side of the above equation. Note that, $|A_{i_{1}}+A_{i_{2}}|=|A_{i_{1}}|+|A_{i_{2}}|-2\cdot|supp(A_{i_{1}})\cap supp(A_{i_{2}})|\geq 2\delta-2,$ owing to the fact that $|supp(A_{i_{1}})\cap supp(A_{i_{2}})|=|\{i: A_{i_{1},i}=1,\; A_{i_{2},i}=1 \}|\leq 1$ as $G$ is $K_{2,2}-$free.

For a general $r$, we have:
\begin{align*}
|A_{i_{1}}+A_{i_{2}}+...+A_{i_{r}}|&= |\sum_{k=1}^{r-1}A_{i_{k}}|+|A_{i_{r}}|-2\cdot |supp(A_{i_{1}}+...+A_{i_{r-1}})\cap supp(A_{i_{r}})|\\
&\geq  |\sum_{k=1}^{r-1}A_{i_{k}}|+|A_{i_{r}}|-2\cdot |\bigcup_{k=1}^{r-1} (supp(A_{i_{k}})\cap supp(A_{i_{r}}))|\\ 
&\geq |\sum_{k=1}^{r-1}A_{i_{k}}|+|A_{i_{r}}|-2\cdot \sum_{k=1}^{r-1}|supp(A_{i_{k}})\cap supp(A_{i_{r}})|,\;  \; \text{(Union Bound)} \\ 
&\geq |\sum_{k=1}^{r-1}A_{i_{k}}|+|A_{i_{r}}|-2\cdot (r-1)\; \;  \text{($G$ is $K_{2,2}$-free)} \\
&\geq |\sum_{k=1}^{r-1}A_{i_{k}}|+\delta-2\cdot (r-1)
\end{align*}

The above set of inequalities can be solved to obtain the inequality:
$$|\sum_{k=1}^{r}A_{i_{k}}|\geq \sum_{k=1}^{r}(\delta-2\cdot(k-1)).$$

Substituting $r=l$, we get:
$$|\sum_{k=1}^{l}A_{i_{k}}|\geq \sum_{k=1}^{l}(\delta-2\cdot(k-1))=\delta+\sum_{k=2}^{l}(\delta-2\cdot(k-1)).$$

Notice that $\delta-2\cdot (k-1)\geq 0,$ for $2\leq k \leq l$, where $l=|b|<\delta/2$.

This implies,
$$\delta\leq |\sum_{k=1}^{l}A_{i_{k}}|=|A_{G}b^{T}|=|a^{T}|<\delta/2,$$
hence giving a contradiction.
\end{proof}
Now, the only missing piece is to ensure for every integer $n$ one can find a $K_{2,2}-$free graph $G$ such that $\delta=\Omega(\sqrt{n})$. We use the following construction of \cite{Erdos1966Algebraic} (see also Lemma 4 in \cite{kapshikar2023hardness}).

\begin{Lemma}\label{lem:C4-free} (Lemma 4, \cite{kapshikar2023hardness})
    Let $n = p^{2}+p+1$ for some prime $p$. Then there is an 
 algorithm that, given $n$, runs in time $O(poly(n))$ and constructs a graph $G$ with $n$ vertices such that $G$ is $K_{2,2}-$free and each vertex of $G$ has degree $p$ or $p+1$.
\end{Lemma}

\begin{proof}[Proof of Lemma \ref{Lemma:sqrtcodedualcode}]
    The result follows from Lemma \ref{lem:graph-based-codes} and Lemma \ref{lem:C4-free} 
\end{proof}

\bibliographystyle{alpha}
\bibliography{distance}

\end{document}